\documentclass{article}
\usepackage{amsthm}
\usepackage{tikz}
\usepackage{amssymb}
\usepackage{amsfonts}
\usepackage{amsmath}
\usepackage[vlined,ruled]{algorithm2e}
\usepackage{graphicx}
\usepackage{hyperref}

\def\F{\ensuremath{\mathcal{F}}} 
\def\N{\ensuremath{\mathbb{N}}}  
\def\K{\ensuremath{\mathbb{F}}}  
\def\I{\ensuremath{\mathcal{I}}} 
\def\R{\ensuremath{\mathcal{R}}} 

\newtheorem{remark}{Remark}
\newtheorem{notation}{Notation}
\newtheorem{definition}{Definition}
\newtheorem{example}{Example}
\newtheorem{proposition}{Proposition}
\newtheorem{corollary}{Corollary}

\newcommand{\spann}[1]{\ensuremath{\langle #1 \rangle}}

\title{On the computation of the M\"obius transform}
\author{Morgan Barbier\footnote{Normandie Univ, UNICAEN, ENSICAEN, CNRS, GREYC, 14000 Caen, France}\\
{\small\texttt{morgan.barbier@ensicaen.fr}}
\and Hayat Cheballah\footnote{Jolibrain, 77 rue pargaminiere 31000 Toulouse - France}\\
{\small\texttt{hayat.cheballah@jolibrain.com}}
\and Jean-Marie Le Bars\footnotemark[1]\\
{\small\texttt{jean-marie.lebars@unicaen.fr}}}
\date{}

\begin{document}
\sloppypar

\maketitle

\begin{abstract}
    The M\"obius transform is a crucial transformation into the Boolean world; it allows to change the Boolean representation between the True Table and Algebraic Normal Form. In this work, we introduce a new algebraic point of view of this transformation based on the polynomial form of Boolean functions. It appears that we can perform a new notion: the M\"obius computation variable by variable and new computation properties. As a consequence, we propose new algorithms which can produce a huge speed up of the M\"obius computation for sub-families of Boolean function. Furthermore we compute directly the M\"obius transformation of some particular Boolean functions. Finally, we show that for some of them the Hamming weight is directly related to the algebraic degree of specific factors.
\end{abstract}


\maketitle


\section{Introduction}
\label{sec:intro}
Numerous studies of Boolean functions have been conducted in various
fields like cryptography and error correcting codes
\cite{carletBook}, Boolean circuits and Boolean Decision Diagram (BDD)
\cite{bryant86}, Boolean logic \cite{boole1848} or constraint
satisfaction problems \cite{creignou2001}.  
There are many ways to represent a Boolean function which depends of the domain. For instance,
on propositional logic one usually uses the conjunctive normal form or
the disjunctive normal form,  while we often use the BDD in Boolean
circuits.

The various criteria of a Boolean function lead us to bring them together 
in numerous classes of Boolean functions which share some set of requirements
and basic operations involved in studies mentioned above consists 
to build a Boolean function in a class 
or to check if a Boolean function belongs to a class.   

Most of the time, practical applications involve several
properties which require different representations. 
For instance, the (algebraic) degree and the (Hamming)
weight are crucial criteria in Cryptography but these basic criteria
are efficiently managed by distinct representations.\\

Indeed, the best representation for the degree is the Algebraic
Normal Form (the characteristic function of monomials), 
while the weight requires the truth table (the characteristic function of
minterms). Both ANF and truth table representation 
require a binary word of length $2^n$, where $n$ 
is the number of variables.

Thus the Reed-Muller decomposition (or expansion) allows
us to perform recursive decomposition , enumeration
and random generation among the degree whereas the Shannon
decomposition (or expansion) does the same task among the weight
\cite{shannon1949} shows the switching network
interpretation of this identity, but Boole will be the first to
mentioned it \cite{boole1854}. 

As its name implies, Reed-Muller decomposition is applied in error correcting codes for Reed-Muller codes\cite{Kasami70}, but also various other fields, for example to implemente circuits with AND/OR gates~\cite{circuit}. 
Furthermore it is often used to construct classes of boolean functions. One example is the  Maiorana-McFarland’s functions where Boolean functions are obtained by expansions of affine functions 
(see~\cite{dillon, mcfarland} for the first studies and~\cite{carletMaiorana} for the use of this class for cryptography). 

Shannon decomposition is very often applied in cryptography, especially 
when we want to maintain a condition over the Hamming weight. However the name is not explicitly mentioned, less specific terms like concatenation or construction are rather used~\cite{siegenthaler,carletBook,carletGouget}.  Furthermore it occurs in various other fields like 
Ordered Binary Decision Diagrams (OBDD)~\cite{OBDD} or Modal Logics~\cite{niveau}.

These decompositions allow us to
rewrite a Boolean function with $n$ variables into two Boolean
functions with $n-1$ variables, while the expansions perform the same acts in reverse, they allows us to build a Boolean function with $n$ variables with two Boolean functions with $n-1$ variables.\\

Since these decompositions appear
to be orthogonal, it seems unreachable to consider them simultaneously
or to perform enumeration or random generation with both
criteria. 

 The M\"obius transform allows to pass from one to the other~\cite{carletBook,guillot,coincidentFunc}. 
The Butterfly algorithm appears as the best known algorithm which performs this transformation. 
It was invented by Gauss in 1805 and Cooley and Tukey independently rediscovered this algorithm for the Fast Fourier Transform (FFT) (Cooley–Tukey FFT algorithm~\cite{gauss,cooley,heiderman84}). 
This is a divide and conquer algorithm which may be implemented in recursive or iterative form.
It has quasi-linear complexity with respect to representation length, $n\; 2^{n-1}$ in term of number 
of XOR operations $\oplus$. However some Boolean functions have compact representation with 
monomials sum or conjunctive or disjunctive normal form and we may expect to get more efficient 
M\"obius transform algorithm for these functions. On the other hand, the M\"obius transform is not 
necessary when we want to answer the two following problems : finding the Hamming weight from the ANF 
and finding the algebraic degree from the truth table. 
The aim of our work is to characterize classes for which 
we have algorithms to answer these two problems more efficient than Butterfly algorithm.


The key ingredient of this work is to manipulate polynomials with M\"obius transform operators 
instead of Boolean functions. Different works in pure Mathematics, as for example complex variable,
provide interesting new results with the polynomial approach \cite{MP_Mobius,OP_Mobius}.
These polynomials are not Boolean functions, they contain 
indeterminates (defined by indices involved in monomials) instead of variables. 
It is possible to go from one world to another by fixing the number of variables of Boolean functions.
We prove that this new approach provides better algorithms to perform M\"obius transform (from ANF to truth table) when we have few monomials in the ANF and monomials of high degree. 
For instance, with $2^{n/2}$ monomials of degree greater or equal to $2^{n/2}$, our method 
has a complexity $2^n$; thus we speed up butterfly algorithm with a $\frac{n}{2}$ factor.




Section~\ref{sec:prelimi} provides the different representations of Boolean functions and exhibit the function to change a representation from another one. In Section~\ref{sec:MobOper}, we discuss on the M\"obius transformation and its first properties.  Section~\ref{sec:MobPoV} is dedicated to reformulate the M\"obius transform for the polynomial form of the Boolean function, thus we deduce faster algorithm to compute it. Finally, Section~\ref{sec:MoComp} shows how to compute directly the M\"obius transform and the Hamming weight of simple and more complicated families of Boolean functions, we conclude with a speed up of greater of 10\% on Achterbahn-128.


\section{Representations of a Boolean functions}
\label{sec:prelimi}
A Boolean function is a mathematical object which is used in different domains: error correcting code, cryptography, constraint satisfaction problems, boolean circuits, etc... Most of time, each of the previous domains use a particular point of view of Boolean functions, thus it exists different representations of Boolean functions. Each point of view make easier to study specific properties of Boolean functions. In this work, we regularly switch between representations. We propose, to give a brief overview of three following representations: Algebraic Normal Form (ANF), truth table, and polynomial form of Boolean functions.

\subsection{Based table representations}

\subsubsection{Monomials and Minterms}
Let $I = \{ i_1, \ldots, i_n\} \subset \N^*$ of cardinality $n$. Let $\F^I_n$ 
be the set of Boolean functions with $n$ variables $x_{i_1}, \ldots, x_{i_n}$
and $\F_n$ be the set of Boolean functions with $n$
variables $x_1,\ldots, x_n$. As we may always switch from $\F^I_n$ to $\F_n$ by renaming the variables, we will conduct our studies over $\F_n$. 

Monomials and minterms play a role of canonical element in
the different writings. 

Let us to denote $x=(x_1,\dots,x_n)$. For any $u = (u_1,\ldots, u_n) \in \K_2^n,$ 
$x^u$ will be denoted the monomial $x_1^{u_1}\dots\, x_n^{u_n}$. The
minterm $M_u$ is the Boolean function with $n$ variables defined by its evaluation
\[
M_u(a) = 
\left\lbrace
\begin{array}{ll}
  1, & \mbox{ if } u = a;\\
  0, & \mbox{ otherwise.}
\end{array}
\right.
\]
Let $u = (u_1,\ldots, u_n) $ and $v = (v_1,\ldots, v_n) \in  \K_2^n$, 
we will write $u \preceq v$, a partial order when $u_i \le v_i$, for any $i \in \{ 1, \ldots, n\}$.

A minterm (resp. a monomial) may be written as a sum of monomials (resp. minterms).
\[
\left\{
\begin{array}{lcl}
M_u & = & \bigoplus_{u \preceq v} x^v;\\
x^u & = & \bigoplus_{u \preceq v} M_v.
\end{array}
\right.
\]

\subsubsection{Characteristic functions of monomials and minterms}
Let $f \in \F_n$ be a Boolean function, $f$ may be viewed as a sum of of minterms
\[ 
f = \bigoplus_{u \in \F_2^n} \theta_u M_u, \text{ with } \theta_u \in \K_2.
\]
Its Truth Table is the characteristic function of minterms, that is:
\[
T(f) = t_1 \ldots\, t_{2^n},
\]
where $t_k = \theta_{u}$, with $k = \sum_{i=1}^n u_i \; 2^{i-1}$.
Moreover, $f$ may be also viewed as a sum of monomials 
\[ 
f = \bigoplus_{u \in \F_2^n} \alpha_u x^u, \text{ with } \alpha_u \in \K_2.
\]
Its ANF (Algebraic Normal Form) is the 
characteristic function of monomials, that is:
\[
A(f)= a_1 \ldots\, a_{2^n}, 
\]
where $a_k = \alpha_{u}$, with 
$k = \sum_{i=1}^n u_i \; 2^{i-1}$.


\begin{example}
  \label{ex:init}
  Let $f= x_1 \oplus x_1 x_2 \in \F_2$ be a Boolean function with two variables. Then its truth table and its ANF are represented by four long bit sequences, and we have:
  \begin{itemize}
  \item $T(f) = 0100$;
  \item $A(f) = 0101$.
  \end{itemize}
\end{example}

Obviously, we may choose in both cases other orders to encode the characteristic function, 
we may for instance permute the order of variables.

\subsection{Polynomial representation}

Let $n \in \N$ and $i_1, \ldots, i_n \in \N$, we will denote by 
$\K_2 [X_{i_1}, \ldots, X_{i_n}]$ 
the set of polynomials over the field $\K_2$ with the indeterminates $X_{i_1}, \ldots, X_{i_n}$. 

\begin{notation}
Let $n \in \N^*$ and $u = (u_1,\ldots, u_n) \in \K_2^n$. 
Recall that $x^u$ is the monomial $ x_{u_1} x_{u_2} \dots x_{u_n}$. 
In order to distinguish a monomial over Boolean functions and a monomial over polynomials, 
we use the respective notations $x^u$ and $X^u$.
\end{notation}

\begin{definition}[Polynomial form]
Let $f \in \F_n$ such that 
$$
f = \bigoplus_{u \in \K_2^n} \alpha_u x^u.
$$
We call the polynomial form of the Boolean function $f$, the polynomial in $\K_2[X_1, \dots, X_n]$:
$$
\pi_n(f) = \sum_{u \in \K_2^n} \alpha_u X^u. 
$$
\end{definition}

Since an indeterminate $X_j \in \{X_{i_1}, \ldots, X_{i_n} \}$ could not occur in $P \in \K_2 [X_{i_1}, \ldots, X_{i_n}]$, 
we have  $\K_2 [X_{i_1}, \ldots, X_{i_n}] \subset  \K_2 [X_{j_1}, \ldots, X_{j_m}]$ 
if $\{ i_1, \ldots, i_n\} \subset \{ j_1, \ldots, j_m \}$. 
Conversely $P \in \K_2 [X_{i_1}, \ldots, X_{i_n}]$ means that any indeterminate $X_j$ which occurs 
in $P$ belongs to $\{X_{i_1}, \ldots,X_{i_n}\}$. Thus the polynomial $X_1 + X_1 X_2$ belongs to $\K_2 [ X_1, X_2]$ but also belongs to $\K_2 [X_1,X_2,X_3]$. We will use the term indeterminate instead of variable to notice that we manipulate formal terms $X_{i_j}$ without notion of evaluation. Moreover, in order to have $\pi_n$ bijective, we define:
$$
\begin{array}{rcl}
    \pi_n : \F_n & \longrightarrow & \K_2[X_1,\dots, X_n]/ (<X_1^2, \dots, X_n^2>)\\
    f   & \longmapsto & \pi_n(f).
\end{array}
$$
Let $P \in \K_2[X_{i_1}, \ldots, X_{i_n}]$, for any $i \in \N$, we will consider the following decomposition 
\[
P = X_i P^0_i + P^1_i, 
\]
where the second part contains all the monomials without the indeterminate $X_i$
and the first one contains all the other monomials. 
Please note that the polynomials $P_i^0$ and $P_i^1$ does not contain the indeterminate $X_i$.
Obviously, the case $P^0_i = 0$ means the indeterminate $X_i$ does not occur in $P$.

\begin{example}[Example~\ref{ex:init} continued]
\label{ex:ANF_PF}
Let $f = x_1 \oplus x_1 x_2 \in \F_2$. We define $f_3 \in \F_3$ and $f_4 \in \F_4$ such that $\pi_2(f) = \pi_3(f_3) = \pi_4(f_4)$. Then
\[
\begin{array}{lcllcl}
\pi_2(f) & = & X_1 + X_1X_2\\
A(f) & = & 0101 & T(f) & = & 0100\\
A(f_3) & = & 01010000 & T(f_3) & = & 01000100\\
A(f_4) & = & 0101000000000000 & T(f_4) & = & 0100010001000100
\end{array}
\]
\end{example}

Although the polynomial form seems to be identical to the ANF, 
we can see in Example~\ref{ex:ANF_PF}, that the size of the ANF representation fixed the number of variables.

\subsection{Differences and similarities between representations}

\subsubsection{Hamming weight and algebraic degree}
 Let $f\in \F_n$ be a Boolean function, we will write $w_H(f)$ the (Hamming) weight of $f$, 
{\it ie} the number of 1 of $T(f)$ and 
$\deg(f)$ the (algebraic) degree of $f$, 
{\it ie} the maximal degree of the monomials in the polynomial or ANF of $f$.

\subsubsection{Shannon and Reed-Muller decompositions}
While the Reed-Muller decomposition is related to the algebraic normal
form, the Shannon one is associated to truth table. Indeed, let $f \in \F_n$ and $i \in \{1, \ldots, n\}$, the {\em Reed-Muller decomposition} in $x_i$, consists in rewriting the Boolean function as
$$f = f_R^0 \oplus x_i f_R^1,$$ 
where $f_R^0, f_R^1 \in \F_{n-1}^{\{1, \ldots, n\}\setminus i}$ and are unique. Clearly, the part $x_i f_R^1$ correspond exactly at all monomials of $f$ where $x_i$ is, and $f_R^0$ the part of $f$ where $x_i$ is not. 
In the particular case where $i = n$, we get $f_R^0, f_R^1 \in \F_{n-1}$, hence we don't have to perform renaming. Furthermore, let $\|$ be the concatenation over words, then
\[ 
A(f) = A(f_R^0) \; \| \;A(f_R^1),
\]
$A(f_S^0)$ (resp. $A(f_S^1)$) contains all the monomials $x^u$ of the ANF of $f$, 
where $u_n = 0$ (resp. $u_n = 1$).

The {\em Shannon decomposition} in $x_i$, consists in
  rewriting the Boolean function as
  $$
  f = (1\oplus x_n)f_S^0 \oplus x_n f_S^1,
  $$
  where $f_S^0, f_S^1 \in \F_{n-1}^{\{1,\ldots, n\}\setminus i}$ and are unique. Clearly the part $(1\oplus x_n)f_S^0$ gives the part of $f$ when $x_i=0$, and $x_i f_S^1$ the part of $f$ when $x_i = 1$. 
  In the particular case where $i = n$, we also get $f_R^0, f_R^1 \in \F_{n-1}$ and we don't have to perform renaming. Furthermore
\[ 
T(f) = T(f_S^0) \; \| \; T(f_S^1),
\]
$T(f_S^0)$ (resp. $T(f_S^1)$) contains all the minterms $M_u$ of the ANF of $f$, 
where $u_n = 0$ (resp. $u_n = 1$).
\begin{remark}
Whether for Reed-Muller or Shannon decomposition, the decomposition in $x_n$ is the only one which allows concatenation. This provides efficient algorithms and avoid induced errors by renaming, which is why we are focusing on this case. 
These decompositions are related to a specific variable but if no variable is defined then the last one $x_n$ is used.
\end{remark}
\begin{remark}
    Let $f \in \F_n$, we have by a trivial identification $f_R^0=f_S^0$
    and $f_R^1 = f_S^0 \oplus f_S^1$.
\end{remark}
\begin{remark}
The Shannon decomposition is the natural decomposition for manipulating 
the minterms since $T(f) = T(f_S^0)\; \|\; T(f_S^1)$.
This trivially implies 
\[
w_H(f) = w_H(f_S^0) + w_H(f_S^1).
\]
On the other hand the Reed-Muller decomposition is the natural decomposition for manipulating 
the monomials; since $A(f) = A(f_R^0)\ \|\ A(f_R^1)$, this implies 
\[
\deg(f) = max \left(\deg(f_R^0), \deg(f_R^1)+1\right).
\]
\end{remark}

\section{M\"obius transform: operator relating the representations}
\label{sec:MobOper}
Since the polynomial and the ANF representation of a Boolean function involving the presence of monomials, it is easy to see these two representations are in direct connection. Moreover, it is a lot more difficult to see that the truth table and the ANF of a Boolean function are connected by a transformation, called the {\em M\"obius transform}. We noted it $\mu$ and is defined by the following bijection
$$\begin{array}{rcl}
    \mu : \F_n & \longleftrightarrow & \F_n\\
    f & \longmapsto & \mu(f),
\end{array}$$
such that, for any $f \in \F_n$ and $a \in \K_2^n$
\begin{equation}\label{mobiusDefinition}
f(a) = \bigoplus_{u \in \K_2^n} \mu(f)(u) a^u.
\end{equation}
The M\"obius transform allows us to compute the truth table representation from ANF one and vice versa. Let $f$ and $g \in \F_n$, the following assertions are equivalent:
\[
\left\{
\begin{array}{lcl}
\mu(f) & = & g;\\
\mu(g) & = & f;\\
A(f) & = & T(g);\\
T(f) & = & A(g).
\end{array}
\right.
\] 

We propose to present a known result \cite[Theorem 5, page 5]{coincidentFunc} in a different usual way. Thus we easy make the link between the Reed-Muller parts $f_R^0$ and $f_R^1$ with the M\"obius transform.

\begin{proposition}\label{prop:Mobius_and_RM}
Let $f \in \F_n$ and $f^0_R,\ f^1_R \in \F_{n-1}$ be the Reed-Muller decomposition of $f$ (in $x_n$), \textit{ie}; $f = f^0_R \oplus x_n f^1_R$. Then
$$\mu(f) = (1 \oplus x_n) \mu(f^0_R) \oplus \mu(f^1_R).$$
\end{proposition}
\begin{proof}
Let $a = (a_1, \ldots, a_{n})$ and $u = (u_1, \ldots, u_{n}) \in \K_2^{n}$, 
$b =  (a_1, \ldots, a_{n-1})$ and $v = (u_1, \ldots, v_{n-1})$. 
we will write $a = b a_n$ and $u = v b_n$. 
It follows $a^u = b^v \; a_n^{u_n}$. Since $a_n^{u_n} = 0$ if and only if $a_n = 0$ and $u_n = 1$, 
we have $a^u = 0$ if $a_n = 0$ and $u_n = 1$ and $a^u = b^v$ otherwise.

The relation~(\ref{mobiusDefinition}) implies
\[
\begin{array}{lcl}
f(b0) & = & \bigoplus_{v \in \K_2^{n-1}} \mu(f)(v0) b^v 0^0 \bigoplus_{v \in \K_2^{n-1}} \mu(f)(v0) b^v 0^1,\\
      & = &  \bigoplus_{v \in \K_2^{n-1}} \mu(f)(v0) b^v;\\
f(b1) & = & \bigoplus_{v \in \K_2^{n-1}} \mu(f)(v1) b^v 1^0 \bigoplus_{v \in \K_2^{n-1}} \mu(f)(v1) b^v 1^1,\\
      & = &  \bigoplus_{v \in \K_2^{n-1}} (\mu(f)(v1) \oplus \mu(f)(v1)) b^v.\\
\end{array}
\]
We deduce
\[
\begin{array}{lcl}
\mu(f)(v0) & = & \mu(f^0_R)(v)\\
\mu(f)(v1) & = & \mu(f^0_R)(v) \oplus \mu(f^1_R)(v)
\end{array}
\]
Thus $\mu(f) = (1 \oplus x_n) \mu(f^0_R) \oplus \mu(f^1_R)$.
\end{proof}

In the following, we propose a new operator, which is related to the M\"obius transform, which is dedicated to manipulate indeterminate one by one.
\begin{definition} \label{def:mux}
Let $f \in \F_n$ and $P = \pi_n(f) \in \K_2[X_{1},\ldots,X_{n}]$ its polynomial form. 
Assume that $P = P^0_i + X_i P^1_i$.
We define the operator $\mu_{X_i}$ by
$$
\mu_{X_i}(P) = P^0_i  + X_i (P^0_i + P^1_i).
$$
In particular, if $i \notin \{ i_1, \ldots, i_n\}$,  
$\mu_{X_i}(P) = (1 + X_i) P$ and if $P = X_i P^1_i$, $\mu_{X_i}(P) = P$. 
\end{definition}

\begin{proposition}\label{commutative}
The operators $\mu_{X_i}$ are commutative, that is
\[
\mu_{X_i} (\mu_{X_j} (P)) = \mu_{X_j} (\mu_{X_i}(P)). 
\]
\end{proposition}
\begin{proof}
Let $P_1, P_2, P_3$ and $P_4$ the four polynomials without the variables $X_i$ and $X_iX_j$ 
such that
\[
P = P_1 + X_i P_2 + X_j P_3 + X_i X_j P_4. 
\]
Thus
\begin{eqnarray*}
\mu_{X_i}(P) & = & P_1 + X_i(P_1 + P_2) + X_j (P_3 + X_i P_4 + X_i P_3);\\
\mu_{X_j}(\mu_{X_i}(P)) & = & P_1 + X_i(P_1+P_2) + X_j(P_1+P_3+X_i(P_1+P_2+P_3+P_4)),\\
& = & P_1 + X_j(P_1+P_3) + X_i(P_1+P_2+X_j(P_1+P_2+P_3+P_4)),\\
& = & \mu_{X_i}(\mu_{X_j}(P)).
\end{eqnarray*}
\end{proof}

\begin{notation}
Let $k \in \N^*$ and $i_1, \ldots, i_k \in \N$. Let $P$ be a polynomial over $\K_2$.
We denote the operator $\mu_{X_{i_1} \ldots X_{i_k}}$ by
\[
\mu_{X_{i_1} \ldots X_{i_k}}(P) = \mu_{X_{i_1}}(\mu_{X_{i_2}}(\ldots \mu_{X_{i_k}}(P)\ldots).
\]
\end{notation}

We may extend the previous Proposition 
for any permutation $\sigma$ of $\{ 1, \ldots, k\}$, 
\[
\mu_{X_{i_1} \ldots X_{i_k}}(P) = \mu_{X_{i_{\sigma(1)}} \ldots X_{i_{\sigma(k)}}}(P).
\]
Hence $\mu_{X_{i_1} \ldots X_{i_k}}(P)$ depends only of the set of indexes $N = \{ i_1, \ldots, i_k \}$.
\begin{notation}
We will write 
$\mu_N(P)$ instead of $\mu_{X_{i_1} \ldots X_{i_k}}(P)$. 
Moreover, let $n \in \N^*$, we will denote by $[n]$ the set $\{1, \ldots, n\}$.
\end{notation}
\begin{example}[Example~\ref{ex:ANF_PF} continued]
  \label{ex:mobius_computation}
  Let $f \in \F_2$ such that its polynomial form is $X_1 + X_1X_2$. Then
  \begin{eqnarray*}
    \mu_{[2]} (X_1 + X_1X_2) & = & \mu_2\left( \mu_1\left( X_1 + X_1X_2\right)\right)\\
    & = & \mu_2\left( X_1 + X_1X_2\right)\\
    & = & X_1X_2 + (1 + X_2) X1\\
    & = & X_1\\
   \mu_{[3]} (X_1 + X_1 X_2) & = & \mu_{\{3\}} (\mu_{[2]} (X_1 + X_1X_2))\\
 & = &  \mu_{\{3\}}(X_1) \\
& = & (1 + X_3) X_1\\
& = & X_1 + X_1 X_3\\
  \mu_{[4]} (X_1 + X_1 X_2) & = & (1 + X_4) ( X_1 + X_1 X_3)\\
& = & X_1 + X_1 X_3 + X_1 X_4 + X_1 X_3 X_4
  \end{eqnarray*}
\end{example}

The following proposition explains how the previous operator is related to the M\"obius transform.
\begin{proposition}
\label{prop:commutative_diagram}
Let $n \in \N^*, f, g \in \F_n$ with polynomial forms $P = \pi_n(f)$ and $Q = \pi_n(g)$.
The following assertions are equivalent:
\begin{itemize}
\item[(a)] $\mu(f) = g$;
\item[(b)] $\mu_{[n]}(P) = Q$.
\end{itemize}

Which yields the following commutative diagram:
$$
\begin{array}{lcccr}
& f & \displaystyle{\overset{\mu}{\longrightarrow}} & g & \\
\pi_n & \downarrow & & \downarrow & \pi_n\\
& P &  \overset{\mu_{[n]}}{\longrightarrow} & Q &
\end{array}
$$
\end{proposition}

\begin{proof}
We only proof that $(a) \implies b$. The other implication is similar.\\ 

We use a induction on $n$. 
For $n=1$, we have by disjunction
\[
\begin{array}{c|c||c|c}
f & P & \mu(f) & \mu_{X_1}(P)\\
\hline
0 & 0 & 0 & 0\\
1 & 1 & 1 \oplus x_1 & 1 + X_1\\
x_1 & X_1 & x_1 & X_1 \\
1 \oplus x_1 & 1 + X_1 & 1 & 1
\end{array}
\]
Since for all $f\in \F_1$, we have $\mu(f) = \mu_{X_1}(P)$, the induction holds for $n = 1$.

Assume now this is true for $n > 1$:
$$\pi_n\left(\mu(f)\right) = \mu_{[n]}\left(\pi_n(f)\right).$$

Let $f\in\F_{n+1}$ be a Boolean function and  $f^0_R, f^1_R \in \F_n$ such that $f = f^0_R \oplus x_{n+1} f_R^1$. Thus with the induction assumption and Proposition~\ref{prop:Mobius_and_RM}:
\begin{eqnarray*}
 \pi_{n+1}\left(\mu(f)\right) & = & (1+X_{n+1})\times \pi_n\left(\mu\left(f_R^0\right)\right)+\pi_n\left(\mu\left(f_R^1\right)\right),\\
  & = & (1+X_{n+1})\times \mu_{[n]}\left(\pi_n\left(f_R^0\right)\right)+\mu_{[n]}\left(\pi_n\left(f_R^1\right)\right),\\
  & = & \mu_{X_{n+1}}\left(\mu_{[n]}\left(\pi_n\left(f_R^0\right)+X_{n+1}\pi_n\left(f_R^1\right)\right)\right),\\
  & = & \mu_{[n+1]}\left(\pi_{n+1}\left(f_R^0+x_{n+1} f_R^1\right)\right),\\
  & = & \mu_{[n+1]}\left(\pi_{n+1}\left(f\right)\right).
\end{eqnarray*}
\end{proof}

Directly, the operator on monomials inherits of M\"obius transform properties.
\begin{proposition}
\label{prop:involutive}
Let $X_i$ be an indeterminate. The automorphism $\mu_{X_i}$ is involutive:
$$\mu_{X_i}^2 = id.$$
\end{proposition}
\begin{proof} 
Let $P = P^0_i + X_i P^1_i$ the Reed-Muller decomposition of polynomial $P$, we denote $Q = \mu_{X_i}(P)$. 
By definition of $\mu_{X_i}$, $Q = P^0_i + X_i (P^0_i + P^1_i)$, thus 
$\mu_{X_i}(Q)= P^0 + X_i (P^0_i + P^0_i + P^1) = P$.
\end{proof}
Propositions~\ref{prop:commutative_diagram} and~\ref{prop:involutive} imply the Corollary below
\begin{corollary}\label{cor:involutive}
Let $N \subset \N$ be a subset, then $\mu_N$ satisfies
\[
\mu_N^2 = id.
\]
\end{corollary} 
Let $f \in \F_n$ be a Boolean function and $P = \pi_n(f)$ its polynomial form; 
Corollary~\ref{cor:involutive} provides an alternative proof that 
$\mu$ is an involutive automorphism, 
since combined with Proposition~\ref{prop:commutative_diagram} it implies $\mu_{[n]}^2 = id$.  
\begin{notation}
Let $I \subset [n]$, we define $M^I = \prod_{i \in [n] \setminus I} \left(1 + X^i\right)$ 
which is the polynomial form of the minterm $M_u$, where $I = I_u$.
\end{notation}
\begin{proposition} \label{prop:monomialMinterm}
Let $I \subset [n]$, then
$$
\mu_{[n]}(X^I) = 
X^I \ \times \ \prod_{i \in [n] \setminus I} \left(1 + X^i\right) = M_I.
$$
Moreover since $\mu_{[n]}$ is an involutive function
$
\mu_{[n]}(M_I) = X^I.
$
\end{proposition}
\begin{proof}
Thanks to Definition~\ref{def:mux}, we obtain by recurrence
\begin{eqnarray*}
\mu_{[n]}(X^I) & = & X^I \ \times \ \mu_{\lbrace X_i\ :\ i\in[n]\setminus I\rbrace}(X^I)\\
 & = & X^I \ \times \ \prod_{i \in [n]\setminus I} \left(1 + X^i \right).\\
\end{eqnarray*}
 And finally Proposition~\ref{prop:involutive} holds the last statement.
\end{proof}
This provide an alternative proof $\mu(x^u) = M^u$ and $\mu(M^u) = x^u$.

\section{A new method to compute the M\"obius transform}
\label{sec:MobPoV}
We have introduced the M\"obius transform over polynomials and show that it is possible to perform the computations in several steps 
with various orders thanks to the partial operators $\mu_{X_i}$. 
We propose to firstly reformulate the M\"obius transform over polynomials in order to introduce two new algorithms based on this reformulation.

\subsection{Reformulation of M\"obius transform}

To introduce our reformulation let us to present a new operator given in the following definition.
\begin{definition}[Exclusive multiplication]
Let $P$ be a polynomial over $\K_2$ and $i \in \N$, $P_i^0$ and $P_i^1$ such that 
$P = P_i^0 + X_i P_i^1$.   
We define the {\em exclusive multiplication}, noted $\otimes$, as
\[
P \otimes X_i = X_i P_i^0.
\]
Let $I$ be a finite subset of $\N$, we generalize the definition for a monomial $X^I$.
\[
P\otimes X^I = X^I P_{\not\vert \; I},
\]
where  $P_{\not\vert \; I}$ is formed with the monomials of $P$ which contain no variables $X_i$, with $i \in I$.
  
We may now generalize for any polynomial $Q$.
Let $\mathcal I$ be a set of finite subsets of $\N$ and $Q = \sum_{I \in \mathcal I} X^I$,
\[
P \otimes Q = \sum_{I \in \mathcal I} P \otimes X^I.  
\]
\end{definition}
\begin{proposition}
\label{prop:commut}
Let $I = \{ i_1,\ldots, i_k \}$ be a finite subset of $\N$. 
\[
P \otimes X^I = (\ldots (P \otimes X_{i_1}) \otimes X_{i_2}) \otimes \ldots X_{i_k}). 
\]

\end{proposition}

Thanks to the previous definition, we can reformulate the M\"obius transform of the Boolean function as a multiplication; this is the result of the following proposition.
\begin{proposition}
  \label{prop:MobiusAsMult}
Let $P$ be a polynomial over $\K_2^n$ and $i \in \N$.
\[
P \otimes (1 + X_i) = \mu_{X_i}(P).
\]
\end{proposition}
\begin{proof}
Let $P_i^0$ and $P_i^1$ such that $P = P_i^0 + X_i P_i^1$.   
$$
\begin{array}{lcl}
P \otimes (1 + X_i) & = & P + P \otimes X_i\\
& = & P + X_i P^0_i\\
& = & P^0_i + X_i (P^0_i + P^1_i)\\
& = & \mu_{X_i}(P).
\end{array}
$$
\end{proof}

Thanks to the previous results, we obtain the following corollary, which supplies a new reformulation of the M\"obius transform.
\begin{corollary}
\label{cor:reformulation}
Let $P \in \K_2[X_1,\dots, X_n]/(<X_1^2, \dots,X_n^2>)$. Then 
  $$
  \mu_{[n]}(P) = P \otimes \prod_{i=1}^n (1 + X_i).
  $$
\end{corollary}

\begin{proposition}
\label{prop:distriOperator}
Let $P$ be a polynomial over $\K_2$ and $i$ and $j \in \N$.
$$
(P \otimes (1 + X_i)) \otimes (1 + X_j) = P \otimes ((1 + X_i) (1 + X_j)). 
$$
\end{proposition}
\proof
Let $P = P_1 + X_i P_2 + X_j P_3 + X_i X_j P_4$.
By Proposition~\ref{prop:MobiusAsMult},
$$
\begin{array}{lcl}
P \otimes (1 + X_i) & = & P_1 + X_j P_3 + X_i(P_1 + X_j P_3 + P_2 + X_j P_4)\\
                    & = & (P_1 + X_i P_1 + X_i P_2) + X_j (P_3 + X_i P_3 + X_i P_4)\\
(P \otimes (1 + X_i)) \otimes (1 + X_j) & = & P_1 + X_i P_1 + X_i P_2 + X_j (P_1 + X_i P_1 + X_i P_2 + P_3 + X_i P_3 + X_i P_4)\\
  & = & P_1 + X_i (P_1 + P_2) + X_j (P_1 + P_3) + X_i X_j (P_1 + P_2 + P_3 + P_4)\\
P \otimes  ((1 + X_i) (1 + X_j)) & = & P \otimes (1 + X_i + X_j + X_i X_j)\\
       & = & P_1 + P_2 + X_j P_3 + X_i X_j P_4 + X_i P_1 \\
&& + X_i X_j P_3 + X_j P_1 + X_i X_j P_2 + X_i X_j P_1\\
   & = &  (P \otimes (1 + X_i)) \otimes (1 + X_j)
\end{array}
$$
\begin{corollary}
Let $P \in \K_2[X_1,\dots, X_n]/(<X_1^2, \dots,X_n^2>)$. Then 
  $$
  \mu_{[n]}(P) = P \otimes (1 + X_1) \otimes (1 + X_2)\ldots \otimes (1 + X_n).
  $$
\end{corollary}

Now, we propose to build an algebraic structure such that the
exclusive multiplication becomes the canonical multiplication in this new structure. Thus we have to create, an algebraic structure such that all monomials containing square indeterminates are projected on zero. We naturally researched a ring which is quotiented by an ideal which represent all these monomials. Thus we obtain the following proposition.
\begin{proposition}
  Let $\I_n$ be the ideal of $\K_2[X_1, \dots, X_n]$ spanned by all the
  indeterminates with a power of two, that is 
  $$
  \I_n = \spann{X_1^2, \cdots,X_n^2}.
  $$
  Then the exclusive multiplication is the natural multiplication in
  the ring
  $$
  \R_n = \K_2[X_1,\dots, X_n] / \I_n.
  $$
\end{proposition}
\begin{proof}
    We propose to prove by inclusion that the ideal $\I_n$ is exactly all monomial with at least a square indeterminates.

    Since $\I_n$ is an ideal, thus by the stability property, we have:
    $$
    \forall a \in \I_n,\ \forall P\in \K_2[X_1, \dots, X_n],\ a . P \in \I_n;
    $$
    thus all monomials containing a square indeterminate is into the ideal $\I_n$.
    
    Let $a \in \I_n$ be an element such that it does not contain any square indeterminate. Since $\I_n$ is spanned by $X_1^2, \cdots X_n^2$, then it exists $a_1, \cdots, a_n \in \K_2[X_1,\cdots,X_n]$ such that:
    $$
    a = \sum_{i=1}^{n}a_i . X_i^2.
    $$
    Since $a$ does not contain any square indeterminate then $\forall i \in \lbrace 1, \cdots, n \rbrace,\ a_i = 0;$ thus $a=0$. We obtain the statement.
    \end{proof}
\begin{corollary}
  Let $f \in \F_n$ be a Boolean function, then the computation of its
  M\"obius transform is only a multiplication on $\R_n$.
\end{corollary}
Thus we reformulate the M\"obius transform such that it is equivalent to canonical multiplication into the quotient ring $\R_n$.
\begin{example}\label{ex:MobiusByMult}
[Example~\ref{ex:mobius_computation} continued]
  With this reformulation, let us to compute again the M\"obius computation of the two variables Boolean function defined by its polynomial form $P = X_1 + X_1X_2$.
  \begin{eqnarray*}
    \mu_{[2]}(P) & = & (X_1+X_1X_2) \otimes (1+X_1) \otimes (1+X_2)\\
    & = & (X_1 + X_1X_2) \otimes (1+X_2)\\
    & = & X_1 + X_1X_2 + X_1X_2\\
    & = & X_1.
  \end{eqnarray*}
  We find exactly the same result that in Example~\ref{ex:mobius_computation}.
\end{example}
\begin{proposition}
  The exclusive multiplication is commutative.
\end{proposition}
\begin{proof}
  The exclusive multiplication is only the canonical multiplication in $\R_n$, moreover $\K_2[X_1, \dots, X_n]$ is a commutative ring, then $\R_n$, that is the exclusive multiplication, also is. 
\end{proof}

\subsection{Algorithms to compute the M\"obius transform}


We propose in this Section an algorithm which compute the M\"obius transform 
with the multiplication $\otimes$. Firstly we will see that it is exactly the same than the iterative 
version of Butterfly algorithm when the algorithm is applied on 
a $2^n$ long bit vector which encodes which monomials occur in $P$ 
(which corresponds to the ANF of $\pi^{-1}_n(P)$). Thus the complexity is $n2^{n-1}$. 
Secondly, we consider $P$ as a list of monomials. 
In this case, we show that this algorithm is better than Butterfly algorithm over large classes 
of Boolean functions.\\

In the first hand, we propose to revisit the Butterfly algorithm and recall a previous improvement. And the other hand, we propose new algorithms from our previous results.

\subsubsection{Butterfly algorithm}

There exists a simple divide-and-conquer butterfly algorithm to perform the M\"obius transform, called 
the Fast M\"obius Transform. We work over $A$, a vector of size $2^n$ which encodes the ANF of 
a Boolean function $f$. Algorithm~\ref{algo:recursivebutterfly} gives the recursive version of the 
Fast M\"obius Transform.

\begin{algorithm}[ht]
  \caption{\label{algo:recursivebutterfly}Recursive butterfly algorithm \textsl{RBM}{(A,n)}}
  \KwIn{$A$ be the ANF (or truth table) of a Boolean function with $n$ variables.}
  \KwOut{the truth table (or ANF) corresponding to $A$.}
  \BlankLine
 \If{$n = 1$}{
   \If{$A = 00$ or $A=01$}{
     \KwRet{$A$}}
     \If{$A = 10$}{
       \KwRet{$11$}}
     \If{$A = 11$}{
       \KwRet{$10$}}
}
\Else{ 
 $A^0 \gets RBM(A[0]\ldots A[2^{n-1}-1],n-1)$\\
 $A^1 \gets RBM(A[2^{n-1}]\ldots A[2^n-1],n-1)$
}
\For{$i = 0$  \KwTo $2^{n-1}-1$}
{$A^1[i] \gets A^1[i] \oplus A^0[i]$}
\KwRet{$A^0||A^1$}
\end{algorithm}

We may directly apply the modifications over $A = A(f)$ without recursive calls. 
For $i =1$ to $n$, we split the string $A$ in $2^{n-i}$ pairs of strings $(A_1,A_2)$ of size $2^{i-1}$ 
and we replace $A_2$ by $A_1 \oplus A_2$, where $\oplus$ is here the bit-wise modulo $2$ sum. Thus it provides a butterfly algorithm working with the memory in place; that is no need extra memory and copy results. It result the Algorithm~\ref{algo:iterativebutterfly} which gives this iterative version of the Fast M\"obius Transform. It is quite the same algorithm introduced in \cite{NNF}, replacing plus operation by XOR.

\begin{algorithm}[ht]
  \caption{\label{algo:iterativebutterfly}Iterative butterfly algorithm \textsl{IBM}{(A,n)}}
  \KwIn{$A$ be the ANF (or truth table) of a Boolean function with $n$ variables.}
  \KwOut{the truth table (or ANF) corresponding to $A$.}
  \BlankLine
\For{$i = 1$ \KwTo $n$}{
 \For{$k = 0$ \KwTo $2^{n-i}-1$}{
  \For{$l = 0$ \KwTo $2^{i-1}-1$}{
      {$A[k*2^i+l+2^{i}] \gets A[k*2^i+l+2^{i}] \oplus A[k*2^i+l]$}
}}}
\KwRet{$A$}
\end{algorithm}

\subsubsection{Optimisation by isolated monomials}

In 2012, Calik Cagdas and Doganaksoy Ali, compute the Hamming weight of Boolean functions from the ANF \cite{CalikDog}. More exactly, a deep reading of this work shows that they compute the Hamming weight of Boolean functions from its polynomial form. Moreover, it provides a new algorithm which can be faster than the butterfly one over a subclass of Boolean function. The previous subclass is mainly defined by they called isolated monomials. That is they rewrite the polynomial form in isolating a monomial, and they take advantage to compute the Hamming weight, their method can be fully detailed in \cite[Algo.~4.1]{CalikDog}. An implementation is even available in \cite{soft_anf2weight}.

\subsubsection{Algorithm with the exclusive multiplication}

From Corollary~\ref{cor:reformulation}, we obtain directly the following algorithm to compute the M\"obius transform.

\begin{algorithm}[ht]
  \caption{\label{algo:reformulation}
    M\"obius transformation by the exclusive multiplication.}
  \KwIn{$P$ be a polynomial form of a Boolean function.}
  \KwOut{$Q$ be the polynomial such that $\mu_{[n]}(P) = Q$.}
  \BlankLine
  $P_0 \gets P$\\
  \For{$i = 1$ \KwTo $n$}{
    $P_i \gets P_{i-1} \otimes \left( 1 + X_i \right)$;
    }
  \KwRet{$P_n$}
\end{algorithm}

We change the point of view of the Algorithm~\ref{algo:reformulation}, in order to make the relation with the Butterfly algorithm.
We encode $P$ by a array $A = A(\pi_n^{-1}(P)))$ of length $2^n$ such that
for each $j = u_1 + u_2 2 \oplus \ldots + u_n 2^{n-1} \in \{0,\ldots, 2^n-1\}$ 
\[
A[j] = 1 \iff \mbox{ the monomial } X^{I_u} \mbox{ occurs in the ANF of } f.
\]
At the step $i$, ($P = P \otimes \left( 1 + X_i \right)$), 
we consider all the $j =  a_1 + a_2 2 \oplus \ldots + a_n 2^{n-1}$ such that $a_i = 0$ 
and we modify the value of $A[j+2^i]$ when $A[j] = 1$.\\

Algorithm~\ref{algo:mobiusiterative} gives the iterative version of Algorithm~\ref{algo:reformulation} 
over the vector $A$ which encodes the monomials. 

\begin{algorithm}[ht]
  \caption{\label{algo:mobiusiterative} Reformulation of Algorithm~\ref{algo:reformulation}.}
  \KwIn{$A$ be the ANF (or truth table) of a Boolean function with $n$ variables.}
  \KwOut{the truth table (or ANF) corresponding to $A$.}
  \BlankLine
 $A \gets A(\pi_n^{-1}(P))$\\
  \For {$i = 1$ \KwTo $n$}{
      \For {every $j = a_1 + a_2 2 \oplus \ldots + a_n 2^{n-1}$, where $a_i = 0$}{
           $A[j+2^i] \gets A[j+2^i] \oplus A[j]$}        
}
    \KwRet{$A$}
\end{algorithm}

We obtain exactly the same that algorithm~\ref{algo:iterativebutterfly}.
Indeed, let $i \in \{ 1,\ldots, n\}$ and 
$j = a_1 + a_2 2 \oplus \ldots + a_n 2^{n-1}$, where $a_i = 0$.
Let $l \in \{0, \ldots, 2^i-1\}$ and 
$k \in \{0,\ldots, 2^{n-i-1}-1\}$ such that
\[
\left\{
\begin{array}{lcl}
l & = &  a_1 + a_2 2+\ldots a_{i-1} 2^{i-2}\\
k & = & a_{i+1} + a_{i+2} 2+\ldots + a_n 2^{n-i-1} 
\end{array}
\right.
\]
It follows $j = l + 2^i k$ and the instruction $A[j+2^i] \gets 1 - A[j+2^i]$ is equivalent to 
$A[k*2^i+l+2^i] \gets 1 - A[k*2^i+l+2^i]$.

\subsubsection{Algorithm for list representation}

In this section, we manipulate Boolean function by its polynomial form given by the list of involved monomials. Hence we can avoid useless computation, as for example a XOR bit with zero. However, this representation suffers an extra memory cost compared to the vector representation.
\begin{proposition}
  \label{prop:naive_list_complexity}
 Let $P \in \R_n$ be a polynomial form of the Boolean function with $n$ variables. 
We denote by $P_i$, $i \in \{1, \ldots, n\}$ the polynomial involved in Algorithm~\ref{algo:reformulation} 
and $N(P_i)$ their number of monomials. Then  Algorithm~\ref{algo:reformulation} uses 
$\sum_{i=1}^n N(P_i)$ XORs.
\end{proposition}
\begin{proof}
This is a direct implication of the equality $P_{i-1} \otimes (1 + X_i) = P_{i-1} + P_{i-1} \otimes X_i$ 
(see Proposition~\ref{prop:MobiusAsMult}).
\end{proof}
With Proposition~\ref{prop:naive_list_complexity}, we note that the number of monomials in the list representation is essential for the complexity. 
\begin{corollary}
Let $N = \max \{ N(P_i) \; | \; i \in \{1, \ldots, n\} \}$. 
Algorithm~\ref{algo:reformulation} uses 
at most $n \; N$ XORs.
\end{corollary}

\begin{notation}
Let $P \in \R_n$ be a polynomial form of the Boolean function with $n$ variables. We denote $\bar{P}$ the polynomial form of the complementary Boolean function associated at polynomial $P$, that is
  $$
  P + \bar{P} = \prod_{i=1}^{n}(1 + X_i).
  $$
\end{notation}
Then we propose the following result in order to improve the complexity of our algorithm.
\begin{proposition}
  \label{prop:complementary_mobius}
  Let $P \in \R_n$ be a polynomial form of the Boolean function with $n$ variables. Then
  $$
  \begin{array}{lcl}
  \mu_{[n]}(\bar{P}) & = & \mu_{[n]}(P) + 1;\\
  \mu_{[n]}(P+1) & = & \mu_{[n]}(P) + \prod_{i=1}^{n}(1 + X_i) = \overline{\mu_{[n]}(P)}.
  \end{array}
  $$
\end{proposition}
\begin{proof}
  Let us to compute
  \begin{eqnarray*}
  \mu_{[n]}(\bar{P}) & = & \left( P + \prod_{i=1}^{n}(1 + X_i)\right) \otimes \prod_{i=1}^{n} (1 + X_i)\\
  & = & \left(P  \otimes \prod_{i=1}^{n} (1 + X_i) \right)+
  \left( \prod_{i=1}^{n}(1 + X_i) \otimes \prod_{i=1}^{n} (1 + X_i)\right)\\
  & = & \mu_{[n]}(P) + 1\\
  \mu_{[n]}(P+1) & = & \left( P + 1 \right) \otimes \prod_{i=1}^{n} (1 + X_i)\\
  & = & \left(P  \otimes \prod_{i=1}^{n} (1 + X_i) \right)+
  \prod_{i=1}^{n} (1 + X_i)\\
  & = & \mu_{[n]}(P) + \prod_{i=1}^{n} (1 + X_i).  
  \end{eqnarray*}
\end{proof}
Then if the list representation of the Boolean function is dense, we can take advantage and work on the complementary, which have a sparse representation. Thus mixing with previous results, we improve the complexity for the list representation.
\begin{corollary}
 Let $P \in \R_n$. We may perform Algorithm~\ref{algo:reformulation} 
with  $\min\left( \sum_{i=1}^n N(P_i), \sum_{i=1}^n N(\bar{P_i})\right)$ XORs.

\end{corollary}

%

%
Proposition~\ref{prop:complementary_mobius} is useful in our context, however this result is not dedicated to our reformulation, it is also true with truth table and ANF.

We remark that the order of the multiplication by the affine
polynomial plays an important role since we involved different polynomials $P_i$ 
when we change the order. To illustrate our claim, we
propose to make again Example~\ref{ex:MobiusByMult} by multiplying
with another order.
\begin{example}[Example~\ref{ex:MobiusByMult} continued]
  \label{ex:orderMult}
  Let $f\in \F_3$ be the Boolean function in Example~\ref{ex:MobiusByMult}, with the list representation we can see that we need only to add 3 monomials, that is
  $$
  P = [X_3, X_1X_2, X_1X_3].
  $$
  After the multiplication by affine polynomials, we obtain
  \begin{eqnarray*}
  P \otimes (1+X_1)  & =  & [ X_3, X_1X_2];\\
  P \otimes (1+X_1)\otimes (1+X_2) & = & [X_3, X_1X_2, X_2X_3];\\
  P \otimes (1+X_1)\otimes (1+X_2) \otimes (1 +X_3)& = & \mu_{[3]}(P) = [X_3, X_1X_2, X_2X_3, X_1X_2X_3].  
  \end{eqnarray*}
  If we process the exclusive multiplication in the different order the number of operation in the list, that is add or remove, will considerably increase:
  \begin{eqnarray*}
  P \otimes (1+X_2)  & =  & [X_3, X_1X_2, X_1X_3, X_2X_3, X_1X_2X_3]\\
  P \otimes (1+X_2)\otimes (1+X_1) & = & [X_3, X_1X_2, X_2X_3]\\
  P \otimes (1+X_2)\otimes (1+X_1) \otimes (1 +X_3)& = & \mu_{[3]}(P) = [X_3, X_1X_2, X_2X_3, X_1X_2X_3].
  \end{eqnarray*}
  We obtain the same result with 5 modifications, while
  Example~\ref{ex:MobiusByMult} obtain the same result with only 3.
\end{example}
We show that in Example~\ref{ex:orderMult} the order of the affine polynomials is really important on the number of list modifications. We propose a strategy to minimize the number of modifications: we propose to multiply
by $(1 + X_{i_0})$, where $i_0$ is the indeterminate which occurs the
most of time in the intern representation. Hence we maximize the number of monomials for which one, we do not perform modification. In this way, we propose Algorithm~\ref{algo:list_reformulation} which manage a good order to perform successive exclusive multiplications to obtain the M\"obius transform. 

\begin{algorithm}[ht]
  \caption{\label{algo:list_reformulation}
    Reformulated M\"obius transformation for the list representation}
  \KwIn{$L$ be the list representation of $f\in\F_n$.}
  \KwOut{$Mu$ be the list representation of the M\"obius transform of $f$.}
  \BlankLine

  $Mu \gets L;$\\
  $O \gets occurrence(L)$;\\
  \For{$i = 1$ \KwTo $n$}{
    $i_0 \gets argmax(O)$;\\
    $Mui \gets Mu$;\\
    \For{$M \in Mu$}{
      \If{not $(X_{i_0} \in M)$}
      {
        \If{$X_{i_0} \in Mu$}
        {
          $remove(Mui, X_{i_0}M)$;\\
          $update(O, X_{i_0}M, -1)$;
        }\Else{
          $add(Mui, X_{i_0}, M)$;\\
          $update(O, X_{i_0}M, 1)$;
        }
      }
    }
    $Mu \gets Mui$;\\
    $O[i_0] \gets -\infty$;
  }
  \KwRet{$Mu$}
\end{algorithm}
Where:
\begin{itemize}
\item $occurrence$ computes a table of size $n$ where the $i$-th component-wise gives the number of occurrences of $X_i$;
\item $remove(L, M)$ and $add(L, M)$ modify the list $L$ with the monomial $M$;
\item $update(O, M, value)$ modifies  the occurrence table $O$ for all variables into the monomial $M$ adding $value$. 
\end{itemize}

In Table~\ref{tab:comparison}, we compare our proposed algorithms with the literature.
We see that the list representation is only valuable for really sparse Boolean functions, or thanks to the complementary property, Proposition~\ref{prop:complementary_mobius}, and really dense ones.
\begin{table}[ht]
  \centering
    \caption{\label{tab:comparison}
    Number of XORs or list modifications needed to compute the M\"obius transform in worst case or for the special case $f =x_3 \otimes x_1x_2 \otimes x_1x_3 \in \F_3$.}
  \begin{tabular}{|c|c|c||c|}
    \cline{2-4}
    \multicolumn{1}{c|}{}& Butterfly & Algorithm in \cite{CalikDog} & Algorithm~\ref{algo:list_reformulation}\\
    \hline
    Complexity & $n2^{n-1}$ & & $\min\left( \sum_{i=1}^n N(P_i), \sum_{i=1}^n N(\bar{P_i})\right)$\\
    \hline
     $x_3 \oplus x_1x_2 \oplus x_1x_3$ & 12 & 10 & 3\\
    \hline
  \end{tabular}
\end{table}

\section{Direct M\"obius computations for some Boolean functions}
\label{sec:MoComp}
We have proposed a reformulation of the M\"obius transform which produces two new algorithms: one for the vector representation and the other for the polynomial form. The worst case of these algorithm happen when a variable $x_i$ does not appear. Hence this section is dedicated to directly compute the M\"obius transform and the Hamming weight of a Boolean function for the worst cases of proposed algorithms.

Please note that for the following propositions, we give the M\"obius transform for some families of Boolean functions. Thus, the computation cost of these Boolean functions is only their Hamming weight for simply write the result into the memory.

\begin{proposition}
\label{prop:monomial}
  Let $I \subset [n]$; then
\[
\mu_{[n]}(X^I) = X^I \; \prod_{j \in [n] \setminus I} (1\oplus X_j),
\]
and we have $w_H(\pi_n^{-1}(X^I)) = 2^{n- |I|}$.
\end{proposition}
\begin{proof}
  It is sufficient to combine Proposition~\ref{prop:monomialMinterm} and Proposition~\ref{prop:commutative_diagram}. This result could be also proved with the relation $M_u = \bigoplus_{u \preceq v} x^v$.
\end{proof}

We consider the following basic algorithm to compute $\mu_{[n]}(P)$ which involves the monomial $X^I$. We began with the word $w=(0,\ldots, 0)$ of length $2^n$; then for each monomial $X^I$, we flip the corresponding bits in $w$, hence the complexity depends on $|\bar{I}|$. For instance, if $P = X^I$, we obtain a complexity $2^{n-|I|}$.

For all $i\in [n]$, we find again that the Boolean functions of $\F_n$ 
given by the polynomial form $X_i$ are balanced functions.

\begin{definition}[Valuation]
  Let $P$ be a polynomial defined over a ring $\mathcal{R}$. The valuation of $P$ is the smallest degree of the set of its monomials.
\end{definition}

\begin{example}
  Let $P(X_1, X_2, X_3) = X_1 + X_2X_3$ and $Q(X_1, X_2, X_3) = 1 + X_2X_3 + X_1X_2X_3$ be polynomials over $\mathbb{F}_2[X_1,X_2,X_3]$, then
  $$
  val(P) = 1,\ val(Q) = 0.
  $$
  Moreover, in order to the valuation has order property, it is frequently assumed that $val(0) = - \infty$.
\end{example}
\begin{proposition}
  Let $P = \sum_{I  \in \mathcal I} X^I \in \R_n$ and $M =|\mathcal I|$. 
Then the M\"obius transform of $P$ and the Hamming weight of $\pi_n^{-1}(P)$ 
can be computed with a complexity $\sum_{I _in \mathcal I} 2^{n-|I|}$, with upper bound  
$M\ .\ 2^{n-val(P)}$.
\end{proposition}
\begin{proof}
Let $P = \pi_n(f)$ and $\mathcal I$ such that $P = \sum_{I \in \mathcal I} X^I$.
\[
\mu_{[n]} (P) = \bigoplus_{I \in \mathcal I} \big( X^I \prod_{j \in [n] \setminus I} (1+X_j) \big).
\]
We conclude by observing that each factor of the sum contains $ 2^{n-|I|} \le 2^{n-val(P)}$ terms. 
\end{proof}
For example, if $val(P) = n/2$ and $M = 2^{n/2}$, we obtain an upper bound of the 
complexity $2^{n/2} \cdot 2^{n/2} = 2^n$ 
which is better than the complexity of butterfly algorithm which is $n 2^{n-1}$.

\begin{proposition}
    \label{prop:XnP}
  Let $P$ be the polynomial form of a Boolean function $f \in \F_{n-1}$. Then M\"obius transform of the polynomial $X_n+P$ with $n$ indeterminates is
  $$
  \mu_{[n]}(X_{n} + P) = \mu_{[n-1]}(P) + X_{n} \mu_{[n-1]}(P+1).
  $$
  Moreover the Boolean function $f'=\pi_n^{-1}(X_n+P)$ is a balanced one, that is:
  $$w_H(f') = 2^{n-1}.$$
\end{proposition}
\begin{proof}
Let us to develop the computation thanks to Definition~\ref{def:mux}:
 \begin{eqnarray*}
 \mu_{[n]}(X_n+P) & = & \mu_{[n]}(X_n) + \mu_{[n]}(P)\\
 & = & X_n \mu_{[n-1]}(1) + (1+X_n) \mu_{[n-1]}(P)\\
 & = & \mu_{[n-1]}(P) + X_n \mu_{[n-1]}\left(P + 1 \right).\\
 \end{eqnarray*}
 
 Moreover, applying Proposition~\ref{prop:complementary_mobius}:
  $$
  \mu_{[n-1]}(P+1) = \mu_{[n-1]}(P) + \prod_{i=1}^{n-1} (1 + X_i) = \overline{\mu_{[n-1]}(P)};$$
  thus
    \begin{eqnarray*}
    w_H(f') & = & w_H(f) + 2^{n-1} - w_H(f);\\
    & = & 2^{n-1}.
    \end{eqnarray*}
 \end{proof}
 
\begin{proposition}
  \label{prop:monomial_sum}
  The M\"obius transform of the sum of all monomials of degree one is the sum of all monomials of odd degree; that is
  $$
  \mu_{[n]}\left( \sum_{i \in [n]} X_i\right) = \sum_{J \subset [n],\text{ st } |J| \text{ is odd}} X^J.
  $$
  Thus $w_H\left(\pi_n^{-1}\left(\sum_{i \in [n]}X_i\right)\right) = 2^{n-1}$.
\end{proposition}
\begin{proof}
\begin{eqnarray*}
\mu_{[n]} \left(\sum_{i \in [n]} X_i\right) & = & \sum_{i \in [n]} \mu_{[n]}(X_i)\\
                         & = & \sum_{i \in [n]} X_i \prod_{j \in [n] \setminus \{ i \}} (1 + X_j)\\
                         & = & \sum_{J \subset [n], |J| \text{ is odd}} X^J.\\
\end{eqnarray*}  
\end{proof}

\begin{remark}
  Let $f = \bigoplus_{i=1}^n x^i \in \F_n$ be the Boolean function which is the sum of all monomials of degree 1. 
Since $w_H(f) = N(\mu_{[n]}(\sum_{i \in [n]} X_i))$, Proposition~\ref{prop:monomial_sum} provides an alternative 
proof that $f$ is a balanced Boolean function.
\end{remark}

The following Proposition shows that we may improve the complexity by a factorization.
\begin{proposition}
  \label{prop:mu_factorization}
  Let $I \subset [n], J \subset [n]$ be two subsets such that $I \cap J = \emptyset$ and $n_1 = |I|$. Let $P \in \R_n$ be a polynomial such that $P = X^I \; \left(\sum_{j \in J} X_j\right)$. Then
  $$
  \mu_{[n]}\left(P \right)= \left( \sum_{I \subset L \subset [n] \setminus J}X^L\right) \left( \sum_{K \subset J,|K| odd} X^K\right);
  $$
    and $w_H\left( \pi_n^{-1}(P)\right) = 2^{n-n_1-1}$.
\end{proposition}
\begin{proof}
  Let $n_2 = |J|$, it follows
\begin{eqnarray*}
  \mu_{[n]}(P) & = & X^I \mu_{[n] \setminus I}\left( \sum_{j\in J} X_j\right),\\
  & = & X^I \; \prod_{k \in [n] \setminus (I \cup J)} (1 + X_k)\ \mu_{J} \left(\sum_{j \in J} X_j\right)\\
  & = & \left( \sum_{I \subset L \subset [n] \setminus J}X^L\right) \left( \sum_{K \subset J,|K| odd} X^K\right).
\end{eqnarray*}
Since $\prod_{k \in [n] \setminus (I \cup J)} (1 + X_k)$ gives $2^{n-n_1-n_2}$ terms and 
$2^{n_2-1}$ subsets of $J$ has a odd cardinality, from Proposition~\ref{prop:monomial_sum}, then the statement is hold.
\end{proof}
\begin{example}
  \label{ex:mu_factorization}
  Let $P = X_1X_2 (X_4+X_5)$ be a polynomial form of a Boolean function with five variables, with calculus made in the previous proof, we directly deduce:
\begin{eqnarray*}
  \mu_{[5]} \left(P\right) & = & X_1 X_2 \times \left(1 + X_3 \right) \times (X_4+ X_5)\\
 & = & X_1X_2X_4 + X_1X_2X_5 + X_1X_2X_3X_4 + X_1X_2X_3X_5.
\end{eqnarray*}

Thus, we can check on this example that $w_H(\pi_5^{-1}(P)) = 4 = 2^{5-2-1}$.
\end{example}

\begin{proposition}
  \label{prop:mu_factorizationPlus}
Let $I_1$ and $I_2 \subset [n], J \subset [n]$ be two subsets such that $I_1 \cap I_2 =  I_1 \cap J = I_2 \cap J = \emptyset$, $|I_1| = n_1$ and $I_2 = n_2$. Let $P \in \R_n$ be a polynomial such that $P = (X^{I_1} + X^{I_2}) \; \left(\sum_{j \in J} X_j\right)$. Then its M\"obius transform is
{\footnotesize{
$$
\left( \sum_{K\subset J, |K|\  odd} X^K\right)
\left(\prod_{k\in [n]\setminus(I_1 \cup I_2 \cup J)} (1+X_k)\right)
\left( X^{I_1} \prod_{k\in I_2}(1+X_k) + X^{I_2} \prod_{k\in I_1}(1+X_k)\right),
$$}}
and $w_H\left( \pi_n^{-1}(P)\right) = 2^{n-n_1-1}+2^{n-n_2-1}-2^{n-(n_1+n_2)}$.
\end{proposition}
\begin{proof}
By Proposition~\ref{prop:mu_factorization}
\begin{eqnarray*}
 \mu_{[n]}(P) & = &\left( \sum_{I_1 \subset L \subset [n] \setminus J}X^L+
\sum_{I_2 \subset L \subset [n] \setminus J}X^L\right) \left( \sum_{K \subset J,|K| odd} X^K\right)\\
& =  & \left( \sum_{I_1 \subset L \subset [n] \setminus J,I_2\nsubseteq I_2}X^L+
\sum_{I_2 \subset L \subset [n] \setminus J, I_1 \nsubseteq I_2}X^L\right) \left( \sum_{K \subset J,|K| odd} X^K\right).
\end{eqnarray*}
Since mutual terms $X^L$ satisfy $I_1 \cup I_2 \subset L \subset [n] \setminus J$. 
$L = (I_1 \cup I_2) \cup L'$, where $L' \subset [n] \setminus (I_1 \cup I_2 \cup J)$. 
Hence we have $2^{n-|J|-n_1-n_2}$ such $L$ subsets. $\{K \subset J,|K| odd\}$ contains $2^{|J|-1}$ 
subsets. 
Therefore, since each mutual term is remove twice, we have 
$2\cdot 2^{n-|J|-n_1-n_2}\cdot 2^{|J|-1} = 2^{n-(n_1+n_2)}$ terms to remove. 
\end{proof}
\begin{remark}
We may generalize this proposition with $k$ subsets $I_1$, \ldots, $I_k$ by using the inclusion/exclusion principle. 
\end{remark}

We can easily see that the Boolean functions defined as Proposition~\ref{prop:mu_factorization} has an even Hamming weight. Moreover, we can notice that the size of second subset $J$ does not act in the Hamming weight.
\begin{example}[Example~\ref{ex:mu_factorization} continued]
  Let $Q = X_1X_2 (X_3+X_4+X_5)$ be a polynomial form of a Boolean function with five variables, we have:
\begin{eqnarray*}
  \mu_{[5]} \left(Q\right) & = & X_1 X_2 \times (X_3+X_4+ X_5 + X_3X_4X_5)\\
  & = & X_1X_2X_3 + X_1X_2X_4 + X_1X_2X_5 + X_1X_2X_3X_4X_5.
\end{eqnarray*}
Thus $w_H(\pi_5^{-1}(Q)) = w_H(\pi_5^{-1}(P)) = 4$.
\end{example}

Another important remark is that the M\"obius transform of indeterminate on set $J$ produces only monomials with odd degree. Thus we can generalize the previous result to the following proposition.
\begin{proposition}
  \label{prop:combination_factorization}
  Let $I,J, I', J' \subset [n]$ be four subsets such that $I \cap J = \emptyset = I' \cap J'$, $I\cup J = [n] = I'\cup J'$ and $n_1 = |I|,\ n_1'= |I'|$, moreover $n_1$ and $n_1'$ has not the same parity. Let $P,P' \in \R_n$ be two polynomials such that $P = X^I \; \left(\sum_{j \in J} X_j\right)$ and $P' = X^{I'} \; \left(\sum_{j \in J'} X_j\right)$. Then
$$
\mu_{[n]}(P+P') = \mu_{[n]}(P) + \mu_{[n]}(P'),
$$
and 
$$
w_H\left( \pi_n^{-1}(P+P')\right) = 2^{n-n_1-1} + 2^{n-n_1'-1}.
$$
\end{proposition}
\begin{proof}
  Since $n_1$ and $n_1'$ has different parity and $[n]\setminus (I \cup J) = \emptyset = [n]\setminus (I' \cup J')$, 
we can't have equal monomials in $\mu_{[n]}(P)$ and $\mu_{[n]}(P')$; then it could not have some vanishing. 
Thus the statement is hold.
\end{proof}

\begin{example}
  Let $f \in \F_5$ be a Boolean function such that its polynomial form is defined as:
\[
\begin{array}{lcl}
Q & = & X_1 X_2 X_3 X_4 + X_1 X_2 X_3 X_5 + X_2 X_4 X_1 + X_2 X_4 X_3 + X_2 X_4 X_5\\
  & = & \underbrace{X_1 X_2 X_3 ( X_4 + X_5)}_{=P} + \underbrace{X_2 X_4 (X_1 + X_3 + X_5)}_{=P'}.\\
\mu_{[n]}(Q) & = &  X_1 X_2 X_3 ( X_4 + X_5) + X_2 X_4 (X_1 + X_3 + X_5 + X_1 X_3 X_5).
\end{array}
\]

Thus $w_H(\pi_5^{-1}(Q)) = 6 = \underbrace{2^{5-3-1}}_{=w_H(\pi_5^{-1}(P))} + \underbrace{2^{5-2-1}}_{=w_H(\pi_5^{-1}(P'))} = 2 + 4$.
\end{example}

We propose another generalization of the Proposition~\ref{prop:mu_factorization}.
\begin{proposition}
    \label{prop:mu_factorization2}
  Let $I, J, K \subset [n]$ be three subsets of $[n]$ such that $I, J, K$ is a partition of $[n]$, then 
  $$
    \mu_{[n]}\left(X^I . \left( X^J + X^K\right)\right) = X^I \left( X^J \prod_{k \in K}\left(1+X^k\right) + X^K \prod_{j \in J}\left(1+X^j\right)\right).
  $$
  Moreover, the Hamming weight of this associated Boolean function is $2^{|J|} + 2^{|K|}$.
\end{proposition}
\begin{proof}
  \begin{eqnarray*}
    \mu_{[n]}\left(X^I . \left( X^J + X^K\right)\right) & = & X^I . \mu_{[n]\setminus I}\left( X^J + X^K\right);\\
    & = & X^I \left( \mu_{[n]\setminus I} \left( X^J\right) + \mu_{[n]\setminus I} \left( X^K\right)\right);\\
    & = & X^I \left( X^J \prod_{k \in K}\left(1+X^k\right) + X^K \prod_{j \in J}\left(1+X^j\right)\right).
  \end{eqnarray*}
\end{proof}
The first consequence of the last proposition, we are able to design balanced Boolean functions directly. Moreover, another direct consequence is that the Hamming weight of a Boolean function does not depend of its degree, but here only of the degree of its factorization.

Finally, we conclude this part with a generalization of the previous proposition.
\begin{proposition}
  \label{prop:mu_factorization3}
  Let $I, J, K \subset [n]$ be three subsets of $[n]$ such that 
  $
  I \cap J = I \cap K = J \cap K = \emptyset.
  $
  We denote $L = I \cup J \cup K$, then $\mu_{[n]}\left(X^I . \left( X^J + X^K\right)\right)$ is
  $$
     \prod_{\ell \in [n]\setminus L} \left(1 + X^\ell \right) X^I \left( X^J \prod_{k \in K}\left(1+X^k\right) + X^K \prod_{j \in J}\left(1+X^j\right)\right).
  $$

  Moreover, the Hamming weight of this Boolean function is $2^{n-|L|} \left(2^{|J|} + 2^{|K|}\right)$.
\end{proposition}
All propositions in this section allows us to give directly the M\"obius transform and the Hamming weight of particular Boolean functions. Other similar propositions could be useful, we introduce the previous ones which seem to be the most helpful.
We have few chances to exploit these propositions for a random Boolean function. However, most of Boolean functions used in practice are not random but design by specific constructions. The following example detail the Boolean function into the design of Achertbahn 128. 

\begin{example}
  Achterbahn 128 is a a synchronous stream cipher algorithm developed by Berndt Gammel, Rainer G\"ottfert and Oliver Kniffler\cite{achterbahn1,achterbahn2}. It involves a Boolean function $f_A$ with $13$ variables, which has good cryptographic properties: balanced, its algebraic degree is $4$, correlation immunity of order $8$, 
nonlinearity $3584$ and algebraic immunity $4$.
The polynomial form of $f_A$ may be written with the following factorization
\[
\begin{array}{l}
X_0+X_1+X_2+X_3+X_4+X_5+X_6+X_7+X_8+X_9+X_{10}+X_{11}+X_{12}+\\
X_0X_5 + X_2(X_{10}X_{11}) + X_6(X_5+X_8+X_{10}+X_{11}+X_{12}) + X_8(X_4+X_7+X_9+X_{10}) +\\ 
X_9(X_{10}+X_{11}+X_{12}) + X_{10}X_{12} + X_{12}X_4 + X_0X_5(X_8+X_{10}+X_{11}+X_{12}) + \\
X_1X_2(X_8+X_{12}) + X_1X_4(X_{10}+X_{11}) + X_1X_9(X_8+X_{10}+X_{11}) +\\
X_2X_4(X_8+X_{10}+X_{11}+X_{12}) + X_2X_7(X_8+X_{12}) + X_2X_8(X_3+\\
X_7+X_{10}+X_{11})+X_3X_8(X_4+X_9) + X_4X_7(X_8+X_{12}) + X_4X_8X_9 + \\
X_4X_{12}(X_3+X_9) + X_5X_6(X_8+X_{10}+X_{11}+X_{12}) +\\
(X_1X_2X_3+X_4X_7X_9)(X_8+X_{12})+(X_1X_2X_7+X_3X_4X_8)(X_8+X_{12})+\\
X_1X_3X_5+X_2X_4X_7)(X_8+X_{12}) + (X_1X_3X_8+X_2X_5X_7)(X_8+X_{12})+\\
(X_1X_7X_9+X_2X_5X_7)(X_8+X_{12}) + (X_1X_5X_7+X_2X_3X_4)(X_8+X_{12})+\\
X_6X_8(X_{10}+X_{11}) + X_6X_{12}(X_{10}X_{11}) + X_8X_9(X_7+X_{10}+X_{11})+\\ 
(X_0X_5X_8+X_1X_4X_{12})(X_{10}+X_{11}) + (X_0X_5X_{12}+X_2X_3X_9)(X_{10}+X_{11})+\\
(X_2X_4X_{12}+X_5X_6X_8)(X_{10}+X_{11}) + (X_1X_9X_{12}+X_2X_4X_8)(X_{10}+X_{11})+\\
(X_1X_8X_9+X_5X_6X_{12})(X_{10}+X_{11})+ (X_1X_4X_8+X_2X_9X_{12})(X_{10}+X_{11})
\end{array}
\]
Butterfly algorithm performs the computation of M\"obius transform in $13\times 2^{12} = 53248$ operations.
By Proposition~\ref{prop:monomial_sum}, the M\"obius transform of the sum of all monomials of degree one is the sum of all monomials of odd degree. Then $\mu_{[13]}(\sum_{i=0}^{12} X_i)$ is compute in 
$2^{12}$ operations. Concerning monomials of degree $2$, we have $7$ terms $P$ of the form 
$X_i \ (\sum_{j \in J} X_j)$, where $i \notin J$. By Proposition~\ref{prop:mu_factorization}, 
each $\mu_{[13]}(P)$ is compute in $2^{11}$ operations. Then for monomials of degree $3$, we have $18$ terms $P$ of the form $X_{i_1}X_{i_2} \ (\sum_{j \in J} X_j)$, where $i_1,i_2 \notin J$. By again Proposition~\ref{prop:mu_factorization}, each $\mu_{[13]}(P)$ is compute in $2^{10}$ operations.
Finally for monomials of degree $4$, we have $12$ terms $(X_{i_1}X_{i_2}X_{i_3}+X_{i_4}X_{i_5}X_{i_6})(X_{j_1}+X_{j_2})$, where $\{i_1,i_2,i_3,i_4,i_5,i_6\} \cap \{j_1,j_2\} = \emptyset$. By Proposition~\ref{prop:mu_factorizationPlus}, each $\mu_{[13]}(P)$ is compute in $2^{10} - 2^7$ operations. 
Hence the total number of operations is $2^{12}+7*2^{11}+18*2^{10}+12*(2^{10}-2^7) = 47616$. 
We gain $5632$ operations, that is a reduction of $10.57\%$, only rewriting $f_A$ and use previous propositions.

\end{example}

\section{Conclusion}
\label{sec:conclusion}
The major contribution of our work is to introduce a polynomial form without reference of a specific Boolean function;
since the indeterminates indicate the variables which occurs in the ANF and not the number of variables. 
Which allow us to give a new point of view of the M\"obius transform and to manipulate Boolean functions of various number of variables via different M\"obius transform operators. We derive from this operators two new algorithms to compute the M\"obius transform, which can be view as a reformulation of the famous Butterfly one. Furthermore, after a deeper study of this reformulation, we provide a new algorithm which have a huge speed up for really sparse or dense polynomials.
 We also explicitly compute the M\"obius transform and Hamming weight for some classes of Boolean functions. Finally, we exhibit a subfamily of Boolean functions for which ones their Hamming weight is directly related to the algebraic degree of specific factors.

\section*{Acknowledgement}
We would like to thank the reviewers for their precious comments.

\bibliographystyle{plain}
\bibliography{biblio}

\end{document}